\documentclass[10pt]{article}
\usepackage[english]{babel} 
\usepackage{amsmath,amsfonts,amsthm,amssymb} 
\usepackage{amssymb}
\usepackage{array}
\usepackage{fullpage}
\usepackage{enumerate}
\usepackage{enumitem}
\usepackage{xstring}
\usepackage{ifthen}
\usepackage{tikz}
\usepackage{tikz-cd} 
\usetikzlibrary{calc}
\usetikzlibrary{positioning}
\usepackage[nodisplayskipstretch]{setspace} 
\usepackage[authoryear]{natbib} 
\usepackage[unicode]{hyperref} 
\usepackage{xcolor} 
\usepackage{bm} 
\usepackage{mathrsfs} 

\usepackage{geometry}
\definecolor{linkcolor}{RGB}{109,71,106}

\definecolor{lam1}{HTML}{127DFF}
\definecolor{lam2}{HTML}{95B204}
\definecolor{lam3}{HTML}{D04133}
\definecolor{lam4}{HTML}{1F61B2}

\definecolor{plot1}{RGB}{52, 101, 28}
\definecolor{plot2}{HTML}{4BAD95}

\hypersetup{
colorlinks=true,
linkcolor=lam3,
urlcolor=plot2,
citecolor=lam2        
}

\newtheorem{thm}{Theorem}

\newtheorem{cor}[thm]{Corollary}

\theoremstyle{def}
\newtheorem{definition}{Definition}

\theoremstyle{remark}
\newtheorem{ex}{Example}
\newtheorem{rmk}{Remark}

\makeatletter
\newcommand{\mylabel}[2]{#2\def\@currentlabel{#2}\label{#1}}
\makeatother

\usepackage{titlesec}

\titleformat*{\section}{\bfseries\centering}
\titleformat*{\subsection}{\bfseries\centering}
\titleformat*{\subsubsection}{\itshape}
\titleformat*{\paragraph}{\large\bfseries\centering}
\titleformat*{\subparagraph}{\large\bfseries\centering}

\titlespacing\section{0pt}{10pt plus 4pt minus 2pt}{4pt plus 0pt minus 2pt}
\titlespacing\subsection{0pt}{8pt plus 4pt minus 2pt}{2pt plus 0pt minus 2pt}
\titlespacing\subsubsection{0pt}{6pt plus 4pt minus 2pt}{2pt plus 0pt minus 2pt}

\renewenvironment{abstract}
 {\small
  \begin{center}
  \bfseries \abstractname\vspace{-.5em}\vspace{0pt}
  \end{center}
  \list{}{%
    \setlength{\leftmargin}{14mm}
    \setlength{\rightmargin}{\leftmargin}%
  }%
  \item\relax}
 {\endlist}
 
 \renewcommand{\blacksquare}{\vrule height7pt width4pt depth1pt}
 \newcommand{\blacksquarethin}{\vrule height7pt width2pt depth1pt}

\def \D{\Delta}
\def \R{\mathbb{R}}
\def \N{\mathbb{N}}
\def \W{\Omega}
\def \w{\omega}

\def \P{\mathbb{P}}

\def \F{\mathcal{F}}
\def \l{\pi}
\def \S{\Sigma}
\def \B{\mathcal{B}}

\def \L{\mathcal{L}}

\renewcommand{\int}{\textup{int}}

\def \s{\operatorname{\succcurlyeq}}
\def \>{\rangle}
\def\<{\langle}
\def \p{\rho}
\def \t{\tau}

\renewcommand{\phi}{\varphi}
\newcommand*\wc{{}\cdot{}}

\begin{document}

\onehalfspacing

\title{\vspace{-6ex}\textsc{Unforeseen Evidence 
}}

\author{Evan Piermont\thanks{Royal Holloway, University of London, Department of Economics, \texttt{evan.piermont@rhul.ac.uk}}}

\maketitle

\begin{abstract}
I propose a normative updating rule, \emph{extended Bayesianism}, for the incorporation of probabilistic information arising from the process of becoming more aware. Extended Bayesianism generalizes standard Bayesian updating to allow the posterior to reside on richer probability space than the prior. I then provide a behavioral characterization of this rule to conclude that a decision maker's subjective expected utility beliefs are consistent with extended Bayesianism.

\bigskip
\noindent \emph{Key words: extended Bayesianism; reverse Bayesianism; conditional expectations.}
\end{abstract}

\setlength{\abovedisplayskip}{3pt}
\setlength{\belowdisplayskip}{3pt}

\newpage

\section{Conditioning on Unforeseen Evidence} Decision makers (DMs) who are \emph{unaware} cannot conceive of, nor articulate, the decision relevant contingencies they are unaware of. Nonetheless, such agents may hold sophisticated probabilistic beliefs regarding those contingencies they \emph{are} aware of. How then should a DM's probabilistic beliefs respond to the discovery of novel contingencies? 
This paper introduces a generalization of Bayesianism called \emph{extended Bayesianism} that restricts how a DM should construct a posterior when she both becomes more aware (perceives more contingencies) and receives evidence (learns an event in the state space has obtained). In particular, the evidence may be \emph{unforeseen}, that is, involve contingencies of which she was perviously unaware.

\begin{ex}
After seeing a doctor for some diagnostic tests, a patient understands there are three possible diseases that could explain his symptoms: diagnoses $A$, $B$, and $C$. The patient looks up the relative frequencies of these diseases and thus maintains a well defined probability distribution over the outcomes of his tests. The doctor then calls the patient and informs him that new research has determined that disease $C$ actually comes in two variants $C1$ and $C2$. He further explains that the $C2$ variant produces unique diagnostic markers not found in the patient and can therefore be ruled out.

The doctor clearly provided information about the patient's health, and this information may intuitively alter the patient's beliefs about possible outcomes. However, the information was \emph{unforeseen}---the event the patient learns,  $\{A,B,C1\}$, was not one he originally considered and therefore not one he had perviously ascribed probability to---precluding the usual Bayesian machinery for updating beliefs.
\hfill $\blacksquarethin$
\end{ex}

In the example, standard Bayesianism does not apply; nonetheless, the logic that underlies Bayesian updating remains and allows us to place reasonable restrictions on the patient's updated belief. Notice that the excluded event ($C2$) was contained entirely in the event $C$. So, the \emph{relative} probabilities between $A$ and $B$ should remain unchanged as the evidence does not favor one event over the other. The \emph{absolute} probabilities of $A$ and $B$, on the other hand, should (weakly) increase since the their complement has now become (weakly) less likely. Notice that if $C$ had been ruled out (rather than $C2$) then these two conditions would characterize the usual method of Bayesian updating (redistributing the probability from $C$ over $A$ and $B$ according to their ex-ante likelihoods). This paper examines the model imposed by these conditions, which characterize Bayesianism in the usual setting, in the more general environment where the DM might receive unforeseen evidence. 

With a bit more formality, assume that the uncertainty relevant to a decision problem can be represented by a state space $\W$. If the DM is unaware of, or otherwise cannot perceive, some decision relevant contingencies, then her perception of the state space will be coarse. Let $\S_0$ and $\S_1$ collect the algebra of events she can perceive at time $0$ and 1, respectively. The notion of becoming more aware is captured by $\S_1$ containing more events than $\S_0$. 

The main definition of this paper, \emph{extended Bayesianism} (EB) relates a prior $\pi_0$, defined over $\S_0$, with a posterior $\pi_1$, defined over $\S_1$, when the DM learns some event $S_1 \in \S_1$. Extended Bayesianism assumes that the prior and posterior are related through the existence of a interim probability $\bar\pi$ defined over $\S_1$ such that (i) $\bar\pi$ is an extension of $\pi_0$ and (ii) $\pi_1$ is the conditional of $\bar\pi$ with respect to the event $S_1$. 

\addtocounter{ex}{-1}
\begin{ex}[continued]
Let $\W = \{\w_A,\w_B,\w_{C1}, \w_{C2} \}$, $\Sigma_0$ be generated by the partition $\{\{\w_A\},\{\w_B\}, \{\w_{C1},\w_{C2}\}\}$ and $\Sigma_1$ by the discrete partition. Let $\pi_0$ be given by $\pi_0(\{\w_A\}) = \frac{1}{2}$ and $\pi_0(\{\w_B\}) = \pi_0(\{\w_{C1},\w_{C2}\} = \frac{1}{4}$. Finally, let $\pi_1(\{\w_1\}) = \frac{4}{7}$, $\pi_1(\{\w_B\}) = \frac{2}{7}$, $\pi_1(\{\w_{C1}\}) = \frac{1}{7}$, and $\pi_1(\{\w_{C2}\}) = 0$. Then $(\pi_0,\pi_1)$ satisfies EB as indicated by the interim probably $\bar\pi$ on $(\W,\Sigma_1)$ given by $\bar\pi(\{\w_1\}) = \frac{1}{2}$, $\bar\pi(\{\w_B\}) = \frac{1}{4}$, $\bar\pi(\{\w_{C1}\}) = \frac{1}{8}$, and $\bar\pi(\{\w_{C2}\}) = \frac18$.
Notice that the relative probabilities of $\w_A$ and $\w_B$ remain constant whereas their absolute probabilities increased. 
\hfill $\blacksquarethin$
\end{ex}

Extended Bayesianism is normatively appealing for several reasons: First, it is a very tight generalization of Bayesianism; the only model consistent with EB when the DM does \emph{not} become more aware (i.e., when $\S_0 = \S_1$) is Bayesianism itself. Second, it has a purely rational interpretation as a two stage updating process whereby first the DM extends her probabilistic assessments to cover newly discovered contingencies and then she conditions her newly constructed beliefs in the standard manner. Of course, she may do these steps simultaneously, so that the interim belief, $\bar\pi$, will be unobservable---I provide an alternative characterization of EB that is observable, in the sense that it places conditions directly on $\pi_0$ and $\pi_1$ without mention of the (possibly hypothetical) interim beliefs.

Finally, EB is minimal, placing restriction on the probabilities of newly discovered events only insofar as what is dictated by Bayesianism itself. To investigate this final motivation in more detail, this paper lays forth a characterization of EB
from the vantage of a DM's betting behavior in the standard decision theoretic environment.

Consider two events $E, F \in \S_0$ such that the DM would prefer to bet $x$ on $F$  than $y$ on $E$ at time 0 (that is, she would prefer to receive $x$ conditional on $F$ (and nothing otherwise) rather than to receive $y$ conditional on $E$).  If the DM then learns that $S_1$ has obtained, and this reverses her preference---she now strictly prefers to bet $y$ on $E$---what can we conclude about the evidence, $S_1$, and its relation to the two events, $E$ and $F$? 

In the standard Bayesian setup, it \emph{must} be that some sub-event $G \subseteq F$ was considered possible---$\pi_0(G) > 0$---but contradicted the evidence---$G \cap S_1 = \emptyset$. Call such an event \emph{discarded} since it was considered before, but not after, the evidence was obtained. 
Then, call a (subjective expected utility) DM \emph{extension consistent} if preference reversals occur only if the event that becomes relatively less appealing contained a discarded sub-event.

\addtocounter{ex}{-1}
\begin{ex}[continued]
Due to the the confusing bureaucracy of the healthcare system in the patient's home country, he must purchase insurance for each disease separately. The patient can insure himself against the cost of treating disease $A$ at cost $x_A \in \R$ (and $B$ and $C$ at cost $y_B$ and $z_C$, respectively). Given his initial probabilistic assessment, at the current rates, the (expected value maximizing) patient is indifferent between insuring against $A$ or $B$ and strictly prefers to insure against $C$. We can write this as $z_C \succ_0 x_A \sim_0 y_B$.

Then, how might the patient's preferences change in the face of the doctors unforeseen exclusion of $C2$? Extension consistency requires that the patient remain indifferent between $x_A$ and $y_B$: $x_A \sim_1 y_B$. However, he may reverse his preference regarding $z_C$, since $C$ contains the discarded event $C2$. Indeed, as evinced by the prior part of this example, it is possible that $\pi_0(C) > \pi_1(C)$, reducing the value of insuring against $C$.

What if, in addition, the patient could jointly insure against $A$ and $B$ at rate $w_{AB}$ and initally strictly preferred this to insuring against only $C$: $w_{AB} \succ_0 
z_C$? Unlike the preference between $z_C$ and $x_A$, the patient's preference between $w_{AB}$ and $z_C$ cannot be reversed by the doctors information since the event $A$ or $B$ does not contain a discarded sub-event. 
\hfill $\blacksquare$
\end{ex}

To see that extension consistency is necessary for standard Bayesianism (when $\Sigma_0 = \Sigma_1$), consider the special case where preferences are reversed over equally sized bets: such a preference reversal indicates that 
$\pi_0(F) \geq \pi_0(E)$ and $\pi_1(E) > \pi_1(F). $
When $\pi_1$ is the conditional of $\pi_0$ with respect to the event $S_1$, the latter inequality is equivalent to  
\begin{equation}
\label{eq:intro2}
\frac{\pi_0(S_1 \cap E)}{\pi_0(S_1)} >  \frac{\pi_0(S_1 \cap F)}{\pi_0(S_1)}
\end{equation}
Multiplying \eqref{eq:intro2} by $\pi_0(S_1)$ and subtracting the resulting inequality from $\pi_0(F) \geq \pi_0(E)$, we see that
$$\pi_0(F \cap S_1^c) > \pi_0(E \cap S_1^c) \geq 0,$$ 
so $G = F \cap S_1^c$ is a discarded sub-event.

This turns out to also be sufficient: A DM is extension consistent if and only if her posterior is a Bayesian update of her prior.
Intuitively, this is because extensions consistency requires that betting preferences are invariant for any two events that are completely consistent with the evidence (subsets of $S_1$), preserving the ratios of probabilities of such events.  Bayesianism can thus be characterized by the stipulation that preference reversals occur only when an event is partially incompatible with the evidence. 

Unlike the definition of Bayesianism, which requires that $\pi_0(S_1)$ is defined, nothing in the definition of extension consistency requires that preferences at time 0 and at time 1 are defined over the same set of events. Extended Bayesianism is the model of beliefs characterized by extension consistency when time 1 preferences are defined over a richer space of uncertainty. 

From this perspective we see the minimality of EB. Whenever two events are \emph{both} completely consistent with the evidence, whether the evidence was foreseen or not, betting behavior across the two events should not change. Preference reversals are allowed but only when the event that becomes relatively less appealing is not completely consistent with the evidence---exactly as in the standard model. 

Because EB places the minimal restrictions on updating, the posterior will not in general be uniquely determined by the prior and the evidence. 
In particular, different extensions of $\pi_0$ may yield different posteriors. In certain cases, discussed below, this multiplicity can become extreme and the restrictions embodied by EB reduce to a (very weak) absolute continuity condition.

I argue that this flexibility is in general philosophically appealing and the cases where the restrictions vanish are illuminating. 
EB constrains the interpretation of novel evidence only insofar as it is related to those contingencies the DM already envisions.\footnote{%
For example, consider the case where the DM discovers that an event $E \in \S_0$ is actually composed of two events $E_1$ and $E_2$ (and suppose further the DM perceived no non-trivial sub-events of $E$ at time 0 and that $E \subseteq S_1$). Now, because $E = E_1 \cup E_2$, EB requires that $\pi_1(E_1) + \pi_1(E_2) = \pi_1(E) = \frac{\pi_0(E)}{\bar\pi(S_1)}$. EB is silent, however, on how the DM distributes the probability between $E_1$ and $E_2$.} Thus, EB, as a model of belief updating, allows for the addition of any other environment-specific criteria for how novel evidence should be interpreted but does not impose such additional restrictions itself. 

Call an event $E \in \S_1$ completely non-measurable if it does not contain any (non-trivial) elements of $\S_0$. If the DM learns a completely non-measurable event then EB imposes no restrictions save absolute continuity: that $\pi_0(F) = 0$ implies $\pi_1(F)=0$ for $F \in S_0$. What seems surprising initially---that Bayesianism has essentially nothing to say in the context of learning completely non-measurable events---makes sense by considering the how such events could be interpreted by the DM.\footnote{I thank an anonymous referee for pointing out this line of enquiry.} 

After learning $S_1$, the DM must contemplate the updated likelihood of events in $\Sigma_0$. For events contained in $S_1^c$, the updated probability must be 0, so her problem concerns the events $\{ F \in \S_0 \mid F \cap S_1 \neq \emptyset \}$. When $S_1$ is completely non-measurable then all such $F$ intersect both the evidence and its complement: we have $F \cap S_1 \neq \emptyset \neq F \cap S_1^c$. Since neither $F \cap S_1$ nor $F \cap S_1^c$ is in $\S_0$, the relative weight DM places on these two events (via $\bar\pi$) is unrestricted---hence the change in probability from $\pi_0(F)$ to $\pi_1(F)$ is likewise unrestricted. In other words, no event in $\{ F \in \S_0 \mid F \cap S_1 \neq \emptyset \}$ implies (is a subset of) or is implied by (is a superset of) $S_1$ leaving the DM's subjective correlation between $S_1$ and the events of $\S_0$ completely unrestricted and allowing for an arbitrary effect of conditioning on $S_1$.

In this sense, completely non-measurable events are unrelated to the contingencies that the DM already envisions. Learning a completely non-measurable event is a paradigm shift; it changes the way the DM perceives \emph{all} previously considered contingencies. Indeed, when $S_1$ is completely non-measurable, there is no non-trivial event $F \in S_0$ such that $F = F \cap S_1$. Such evidence, which is both unanticipated by and unrelated to the prior, can therefore radically change beliefs.

The rest of the paper is structured as follows. The next section introduces the main definition, extended Bayesianism. It also provides an alternative characterization through conditions directly on the beliefs. The characterization is both practical (in the sense that it is observable) and helps to motivate the definition. Next, Section \ref{sec:DT} considers a decision theoretic set up and puts forth a characterization of extended Bayesianism through the lens of betting behavior. Section \ref{sec:inf} generalizes the set up to infinite state spaces and Section \ref{sec:rep} to more than two time periods. Finally, some links to the larger literature are discussed in Section \ref{sec:lit}. An appendix, \ref{sec:DT2}, sketches a model where the decision theoretic objects are language based (rather than state space based), which assuages some worries about observability.  

\section{Bayesianism and Extended Bayesianism}

Let $\W$ denote a finite objective, albeit possibly unobservable, state space. Let $\Sigma_t$, a sigma-algebra on $\W$, represent the events the DM can conceive of at $t \in \{0,1\}$. By nature of the problem, we assume that $\Sigma_0 \subseteq \Sigma_1$. The DM's subjective uncertainty, given her current understanding, is taken to be a probability distribution, $\pi_t$, on the probability space $(\W,\Sigma_t)$. Set 
$S_t = \bigcap \{ E \in \Sigma_t \mid \pi_t(E) = 1 \}$
 to denote the support of $\pi_t$, the smallest event in $\Sigma_t$ with $\pi_t$-probability 1. 
Let $\pi_0^{*}: \S_1 \to [0,1]$ denote the outer measure over $\S_1$ induced by $\pi_0$: $\pi_0^{*}(E) = \inf\{\pi_0(F) \mid F \in \S_0, E \subseteq F\}$. Clearly, if $E \in \S_0$ we have $\pi_0^{*}(E) =\pi_0(E)$.

Before introducing the main definition, which captures updating when when the time 1 beliefs are defined over a richer space of events than time 0 beliefs, we first recall the standard definition of Bayesian updating. 

\begin{definition}
\label{def:b}
If $\S_0 = \S_1$, then say that $(\pi_0, \pi_1)$ satisfies \emph{Bayesianism} if 
\begin{enumerate}
\item[\mylabel{b1}{\textsc{b}}] If $\pi_0(S_1) > 0$, then $\pi_1(E) = \frac{\pi_0(E \cap S_1)}{\pi_0(S_1)}$
for all $E \in \Sigma_1$.
\end{enumerate}
\end{definition}

Often, we think of Bayesianism in relation to having learned that some particular event $E$ obtained, in which case the posterior belief is determined by conditioning the prior belief on $E$. Here, instead, we take a slightly more general notion: stating that the pair $(\pi_0, \pi_1)$ satisfies Bayesianism if \emph{there exists} some event such that posterior is the conditional of the prior with respect to this event (when this is true, the event must of course be $S_1$). 
As usual, Bayesianism places very weak requirements on $\pi_1$ when the conditioning event has $\pi_0$-measure 0. In such cases, the only (and trivial) requirement is that the support of the posterior be the conditioning event.

We now turn our attention to the case where the DM's perception of uncertainty in time 1 is finer than her perception at time 0, so that $\S_0$ is a strict subset of $\S_1$. Here, the restriction \eqref{b1} 
is not in general well defined, since $\pi_0$ does not assess all the relevant subsets (for example if $S_1 \notin\S_0$). The main definition, extended Bayesianism, deals with this by asserting the existence of an extension of $\pi_0$ to the richer space of uncertainty, allowing us to define the analog of the restriction \eqref{b1}.

\begin{definition}
Say that $(\pi_0, \pi_1)$ satisfies \emph{extended Bayesianism} (EB), if there exists a probability distribution $\bar\pi$ on $(\W, \Sigma_1)$ such that
\begin{enumerate}
\item[\mylabel{eb2}{\textsc{eb1}}] $\bar\pi(E) = \pi_0(E)$ for all $E \in \Sigma_0$, and 
\item[\mylabel{eb3}{\textsc{eb2}}] If $\pi_0^{*}(S_1) > 0$, then $\bar\pi(S_1) > 0$ and $\pi_1(E) = \frac{\bar\pi(E \cap S_1)}{\bar\pi(S_1)}$
for all $E \in \Sigma_1$.
\end{enumerate}
Moreover, in such cases, call $\bar\pi$ a \emph{witness} to the fact that $(\pi_0, \pi_1)$ satisfy EB (or that $(\pi_0, \pi_1)$ satisfies EB is \emph{witnessed by} $\bar\pi$).
\end{definition}

An interpretation is as follows: If $(\pi_0,\pi_1)$ satisfies EB it is \emph{as if} $\pi_1$ was constructed by conditioning $\pi_0$ on the event $S_1$. I say `as if' because when  $S_1\notin \Sigma_0$ then the $\pi_0$ probability of $S_1$ is undefined.
However, in this case, we make sense of conditioning by first extending $\pi_0$ to the richer algebra ($\pi_0 \to \bar\pi$) and then constructing $\pi_1$ by conditioning this extension ($\bar\pi \to \pi_1$).

Like Bayesianism, EB places no requirements on the posterior when the conditioning event is considered impossible according to the time 0 beliefs. The notion of impossibility, however, must now permit the conditioning event $S_1$ to be unforeseen, and therefore have no $\pi_0$ probability. An event $E \in \S_1$ is impossible according to $\pi_0$ if there is no way to extend $\pi_0$ to $\S_1$ and assign $E$ positive probability: this is precisely when $\pi_0^{*}(E) = 0$. Thus \eqref{eb3} states that the posterior is the (standard) Bayesian update of some extension of $\pi_0$, and this extension assigns the conditioning event postive probability unless it is impossible to do so. 

EB is a generalization of Bayesianism in the tightest possible way. While EB applies on a more general domain, the only theory consistent with EB when $\S_0 = \S_1$ is Bayesianism itself. 
\begin{rmk}
\label{rmk:reduct}
If $\S_0 = \S_1$, then $(\pi_0, \pi_1)$ satisfies Bayesianism if and only if it satisfies extended Bayesianism. 
\end{rmk}
Thus, EB generalizes Bayesianism to novel circumstances but does not permit behavior which is ever at odds with the classical model.

%

\begin{ex}
Let $\Sigma_0 = \{\emptyset, \W\}$. Then $(\pi_0,\pi_1)$ satisfies EB irrespective of $\pi_1$.
\hfill $\blacksquare$
\end{ex}

\begin{ex}
\label{ex:anythinggoes0}
Let $\pi_0^{*}(S_1) = 0$, then $(\pi_0,\pi_1)$ satisfies EB.
\end{ex}

\begin{ex}
Let $S_1 = S_0$. Then $(\pi_0,\pi_1)$ satisfies EB if and only if $\pi_1$ is an extension of $\pi_0$ to $\S_1$.
\hfill $\blacksquare$
\end{ex}


\subsection{Observability} 
Bayesian updating is the normative benchmark for how probabilistic judgements should respond to the acquisition of new evidence. Unfortunately, in cases where $S_1\notin \Sigma_0$, Bayes' rule cannot be directly verified, as there was no prior belief regarding the likelihood of the conditioning event. The notion of commensurability, below, provides a simple resolution, advancing an observable restriction on $(\pi_0,\pi_1)$ equivalent to extended Bayesianism.


\begin{definition}
Say that $\pi_1$ is \emph{commensurate} to $\pi_0$ if for all $E,F \in \Sigma_0$ with $E \subseteq S_1$
\begin{enumerate}
\item[\mylabel{c1}{\textsc{c1}}] $\pi_0(F) = 0 \implies \pi_1(F) = 0$, and,
\item[\mylabel{c2}{\textsc{c2}}]  $\pi_0(E)\pi_1(F) \leq \pi_1(E)\pi_0(F)$.
\end{enumerate}
\end{definition}

\begin{rmk}
\label{rmk:size}
For all $E \in \Sigma_0$ with $E \subseteq S_1$, $\pi_0(E) \leq \pi_1(E)$. This follows by setting $F$ to $\W$ in \eqref{c2}. If both $E,F \subseteq S_1$, then \eqref{c2} holds with equality---this follows from interchanging the roles of $E$ and $F$. 
\end{rmk}


\begin{thm}
\label{prop:update}
$\pi_1$ is commensurate to $\pi_0$ if and only if $(\pi_0,\pi_1)$ satisfies EB and $\pi_0^{*}(S_1) > 0$.
\end{thm}

Notice that by Remark \ref{rmk:reduct}, commensurability serves as a characterization of standard Bayesianism whenever $\S_0 = \S_1$. While we wait to prove Theorem \ref{prop:update} until the next section, its introduction here provides a vantage by which to interpret the restrictions that underlie Bayesian updating (extended or otherwise). 

\eqref{c1} states that the posterior, $\pi_1$, is absolutely continuous with respect to the prior, $\pi_0$. Whatever was considered impossible before obtaining evidence must still be consider impossible after.\footnote{This is because we are narrowing in on the case where the DM does not receive evidence she considered impossible. When the conditioning event is considered ex-ante impossible, there are no additional restrictions needed to characterize EB; see Example \ref{ex:anythinggoes0}.} The intuition behind \eqref{c2} is best seen when $\pi_1(F) > 0$, which by \eqref{c1} implies that $\pi_0(F)> 0$ as well. In this case, we can rewrite \eqref{c2} as
\begin{equation}
\label{eq:altc2}
\frac{\pi_0(E)}{\pi_0(F)} \leq \frac{\pi_1(E)}{\pi_1(F)}.
\end{equation}

Recall, the conditioning event, or `evidence,' is $S_1$. Call an event \emph{consistent} with the evidence if intersects $S_1$ (so that the event is not ruled out) and \emph{completely consistent} if it is a subset of $S_1$ (so that no part of the event has been ruled out).
Now, \eqref{c2} refers to an event $E \subseteq S_1$ that is completely consistent with the discovered evidence and an event $F$ that is consistent but may or may not be completely consistent---if $F \not\subseteq S_1$ then there are some aspects of (i.e., states in) the event $F$ that contradict the discovered evidence.  \eqref{c2} states that \eqref{eq:altc2} holds with equality when both $E$ and $F$ are completely consistent.  So, when the discovered evidence does not contradict any part of either event, then the \emph{relative} likelihood of the events must remain unchanged. 

This ratio preservation is often the motivation or verbal characterization of Bayesianism. However, when considering the more general environment when $\S_0 \subset \S_1$, we must specifically discuss the equally important case when $F$ is \emph{not} completely consistent with the evidence. In such cases, the ratio between the likelihood of $E$ and $F$ is not preserved, but, it must change in such a way as to make $E$---the event completely consistent with the evidence---relatively more likely. 

Combined, we see that extended Bayesianism (and by reduction, standard Bayesianism) can be summed up as the restriction that relative likelihoods can change \emph{only} when the events contain states contradictory to the evidence, and even then only so as to make completely consistent events more likely. This characterization will be interpreted in a behavioral context in the next section. 

The characterization of EB through commensurability also helps to illuminate when the (extended) Bayesian paradigm has nothing to say about how beliefs should change. 

\begin{definition}
Say that $E \in \S_1$ is \emph{completely non-measurable} if for all $F \subseteq E$, $F \neq \emptyset$, $F \notin \S_0$.
\end{definition}

\begin{cor}
\label{rmk:diff}
Fix $(\pi_0,\pi_1)$ such that $S_1$ is completely non-measurable. Then $(\pi_0,\pi_1)$ satisfies EB and $\pi_0^{*}(S_1) > 0$ if and only if \eqref{c1} holds.
\end{cor}

Corollary \ref{rmk:diff} follows by noticing that \eqref{c2} is trivially satisfied when $S_1$ is completely non-measurable, as there is no event $E \in \S_0$ such that $E \subseteq S_1$. As the following example shows, this is pertinent in the case where $\W$ is a product space. 

\begin{ex}
Let $\W = \W_0 \times \W_1$ be a product space. Let $\S_0$ be generated by the sets of the form $F_0 \times \W_1$ for $F_0 \subseteq \W_0$. Take some $(\pi_0,\pi_1)$ such that $S_1$ is of the form $\W_0 \times E_1$ with $E_1 \subset \W_1$, $E_1 \notin \{\emptyset, \W_1\}$. Then $S_1$ is completely non-measurable, and hence $(\pi_0,\pi_1)$ satisfies EB as long as $\pi$ is absolutely continuous with respect to $\pi_0$.
\end{ex}

Substantiating the interpretation of completely non-measurable events as paradigm shifts, the example shows that, except in trivial cases, when the DM discovers a new dimension to the state space, EB trivializes. Intuitively, this is because the marginal of $\bar\pi$ over $\W_1$ is completely unrestricted, allowing arbitrary correlation between $S_1$ and the events in $\W_0$.

\section{A Behavioral Model} 
\label{sec:DT}

In this section we will consider an expected utility maximizing DM who is choosing over consumption acts. We will see that extended Bayesianism, or equivalently the notion of having commensurate beliefs, is captured naturally through a dynamic consistency axiom.

Towards this, let $\F_t$ denote the set of $\S_t$-measurable acts from $\W$ to $X$, with $X$ a non-trivial convex consumption space with a worst element $w$. In the standard abuse of notation, identify $X$ with the set of constant acts inside of $\F_t$. For any act $f\in \F_t$ and an event $E \in \S_t$, let $f_E$ denote the act in $\F_t$ which coincides with $f$ on $E$ and $w$ otherwise.

We will consider a DM who has a time indexed preference over such acts, given by $\s_t$. It is worth noting that the observability of such preferences is not as straightforward as in the usual model: in particular the DM's perception of the state space at any time is actually a coarsening of the true state space, and so the description of an act (by describing each state) might in-and-of-itself alter the perception of the DM. A consistent interpretation is that the DM's perception is informed by her awareness of various contingencies, and that time $t$ acts must depend only on the contingencies she is aware of at time $t$; this model is sketched in Appendix \ref{sec:DT2}.

We assume that $\s_t$ is a subjective expected utility preference, represented by $(\pi_t,u)$, where $\pi_t$ is a probability distribution in $(\W,\S_t)$ and $u: X \to \R_+$ is a non-trivial, continuous and \emph{time invariant} utility index with $u(w) = 0$. Specifically, $\s_t$ is given by the comparison of the value, for $f \in \F_t$,
\begin{equation}
\label{eq:eu}
V_t(f) = \sum_{x \in f(\W)} \pi_t(f^{-1}(x))u(x).
\end{equation}

Before introducing the behavioral counterpart to commensurability, we must first consider the notion of \emph{discarded} events: events which are ruled out in the face of the discovered evidence. As is standard, call an event $E \in \S_t$ \emph{t-null} if $f_E \sim_t f'_E$ for all $f,f' \in F_t$. Call an event $E \in \S_1$ \emph{discarded} if it is  $1$-null, but for all $F \in \S_0$, if $F \supseteq E$ then $F$ is not 0-null.  

An event $E$ is discarded if the DM considers it relevant in time 0, but irrelevant at time 1 after obtaining evidence. If $E \in \S_0$, then $E$ is discarded if and only if it is 1-null and not 0-null.

\begin{definition}
Call $(\s_0,\s_1)$ \emph{extension consistent} if for all $x,y \in X$ and $E,F \in \Sigma_0$,
$x_F \s_0 y_E$ and $y_E \succ_1 x_F$ implies $F$ contains a discarded sub-event. 
\end{definition}

\begin{thm}
\label{thm:maineq}
Let $\s_0$ and $\s_1$ be represented by the subjective expected utility functionals given by $(\pi_0,u)$ and $(\pi_1,u)$, respectively. Then the following are equivalent:
	\begin{enumerate}
	\item $(\pi_0,\pi_1)$ satisfy extended Bayesianism and $\pi_0^{*}(S_1) > 0$,
	\item $(\s_0,\s_1)$ is extension consistent, and,
	\item$\pi_1$ is commensurate to $\pi_0$. 
	\end{enumerate}
\end{thm}

\begin{proof}
We will show that (1) implies (2) implies (3) implies (1).

Assume (1) and that for some $x,y$ we have
$x_F \s_0 y_E$ and $y_E \succ_1 x_F$.  Assume without loss of generality that $u(x) = \lambda \geq 0$, $u(y) = 1$. By the subjective expected utility representation, we have
$$\lambda\pi_0(F) \geq \pi_0(E) \text{ and }  \pi_1(E) > \lambda\pi_1(F). $$
Let $\bar\pi$ be a witness to the fact that $(\pi_0,\pi_1)$ satisfy Extended Bayesianism. Then, the first inequality is equivalent to  
$$ \lambda\bar\pi(F \cap S_1) + \lambda\bar\pi(F \cap S_1^c) \geq \bar\pi(E \cap S_1) + \bar\pi(E \cap S_1^c)$$ and by \eqref{eb3}, since $\pi_0^{*}(S_1) > 0$, the second to
$$ \frac{\bar\pi(S_1 \cap E)}{\bar\pi(S_1)} >  \lambda\frac{\bar\pi(S_1 \cap F)}{\bar\pi(S_1)}.$$
Multiplying the later inequality by $\bar\pi(S_1) > 0$ and subtracting the result from the former inequality, we obtain:
$$\lambda\bar\pi(F \cap S_1^c) > \bar\pi(E \cap S_1^c) \geq 0$$ 
Clearly, $(F \cap S_1^c) \subseteq S_1^c$ is 1-null. Moreover, for any $G \in \S_0$ with $(F\cap S_1^c) \subseteq G$, we have $\pi_0(G) = \bar\pi(G) \geq \bar\pi(F \cap S_1^c) > 0$ so $G$ is not 0-null. Hence, $(F\cap S_1^c)$ is a discarded event contained if $F$. 

Now, assume (2). Let $E,F \in \S_0$ with $E \subseteq S_1$. First, assume by way of contradiction that for some $F$, $\pi_0(F) = 0$ but $\pi_1(F) > 0$. Let $x \in X$ with $x \succ_1 w$ so by the representation we have $x_\emptyset \s_0 x_F$ and $x_F \succ_1 x_\emptyset$; so by extension consistency $\emptyset$ contains a discarded (hence non-0-null) event, a clear contradiction. \eqref{c1} holds.

Now, without loss of generality let $x,y \in X$ be such that $u(x) = \pi_0(F)$ and $u(y) = \pi_0(E)$ (since $X$ is convex, it is connected, and therefore its image under the continuous $u$, $u(X)$, is a connected subset of $\R_+$. This image can be scaled via an affine transform to contain $[0,1]$).
 Then by the representation we have $x_E \sim_0 y_F$. Now since $E \subseteq S_1$, $E$ cannot contain a discarded event, and therefore, by extension consistency, it must be that $x_E \s_1 y_F$. The representation then implies that 
\begin{equation}
\label{eq:c2inproof}
\pi_1(F)\pi_0(E) \leq \pi_1(E)\pi_0(F).
\end{equation}
\eqref{c2} is satisfied, and therefore $\pi_1$ is commensurate to $\pi_0$.

 That (3) implies (1) is is a corollary to the more general Theorem \ref{prop:updateinf}, noting that in finite state spaces, \eqref{c3} is vacuous.
\end{proof}

\section{Infinite State Spaces}
\label{sec:inf}

This section extends the criterion of commensurability to capture extended Bayesianism in countably infinite state spaces. The more subtle complications which arise when considering arbitrary (uncountable) state-spaces are discussed at the end of the section.  
To see why the situation is complicated by infinite state spaces, consider the following example.

\begin{ex}
\label{ex:4}
Let $\Omega = \mathbb{N}\times\{A,B\}$ with $\Sigma_0$ generated by $\N$ and $\Sigma_1$ by the discrete partition. Set $\pi_0(E_0) = \frac12$ and $\pi_0(E_n) = 3^{-n}$ for $n > 0$. Set $\pi_1(E_{0A}) = \pi_1(E_{nB}) = 0$ and $\pi_0(E_{nA}) = 2^{-n}$ for all $n > 0$ (see Figure \ref{fig:ex}). Then $(\pi_0,\pi_1)$ satisfies \eqref{c1} and \eqref{c2} but does not satisfy EB.

To see that $(\pi_0,\pi_1)$ does not satisfy EB, consider the event $S_1 = \bigcup_n E_{nA}$. If $(\pi_0,\pi_1)$ did satisfy EB with witness $\bar\pi$, then what would $\bar\pi(S_1)$ be? From \eqref{eb3}, we have
$$ \frac1{2^n} = \pi_1(E_{n}) = \frac{\bar\pi(S_1 \cap E_{n})}{\bar\pi(S_1)} = \frac{\bar\pi(E_{nA})}{\bar\pi(S_1)} \leq \frac{\frac{1}{3^n}}{\bar\pi(S_1)}$$ 
Therefore, whatever $\bar\pi(S_1)$ is, it must be less than $\frac{2^n}{3^n}$, which, by tending to 0, implies $\bar\pi(S_1) = 0$, violating \eqref{eb3}.
\hfill $\blacksquare$
\end{ex}

\begin{figure}
\centering
\begin{tikzpicture}[scale=.6]

\newcommand*{\GetListMember}[2]{%
    \edef\dotheloop{%
    \noexpand\foreach \noexpand\a [count=\noexpand\i] in {#1} {%
        \noexpand\IfEq{\noexpand\i}{#2}{\noexpand\a\noexpand\breakforeach}{}%
    }}%
    \dotheloop
    \par%
}%

\node at (17.5,0) {$\ldots$};
\node at (17.5,2) {$\ldots$};
\node at (.5,.5) {$\Sigma_1$};
\node at (.5,2.5) {$\Sigma_0$};

\def\probsone{0,$\frac12$,$\frac14$,$\frac18$}
\def\probszero{$\frac12$,$\frac13$,$\frac19$, $\frac1{27}$}

\foreach \y in {1,2,3,4}{%
\pgfmathsetmacro\z{int(\y-1)}
\draw[draw=black,] (-3 + 4*\y,0) rectangle ++(2,1) node[pos=.5] (A\y) {\tiny $\pi_1{=}$\GetListMember{\probsone}{\y}};
\draw[draw=black, fill=gray, fill opacity=.2] (-1 + 4*\y,0) rectangle ++(2,1) node[pos=.5, opacity=1] (B\y) {\tiny $\pi_1{=}$0};
\node[above = .1cm and 0cm of A\y] (e) {$E_{\z A}$};
\node[above = .1cm and 0cm of B\y] (e) {$E_{\z B}$};
}
\draw[draw=black, fill=gray, opacity=.2] (-3 + 4,0) rectangle ++(2,1);
\foreach \y in {1,2,3,4}{%
\pgfmathsetmacro\z{int(\y-1)}
\draw[draw=black] (-3 + 4*\y,2) rectangle ++(4,1) node[pos=.5] (\y) {\tiny $\pi_0{=}$\GetListMember{\probszero}{\y}};
\node[above = .1cm and 0cm of \y] (e) {$E_{\z}$};
}
\end{tikzpicture}
\caption{A visual representation of the state space from Example \ref{ex:4}.}
\label{fig:ex}
\end{figure}

The following additional restriction rules out cases like Example 4, where the ratio of events tends to 0.

\begin{definition}
Say that $\pi_1$ is \emph{boundedly commensurate} to $\pi_0$ if it is commensurate and
\begin{enumerate}
\item[\mylabel{c3}{\textsc{c3}}] $\inf\limits_{\substack{F \in \S_0, \\ \pi_1(F)>0}} \frac{\pi_0(F)}{\pi_1(F)} > 0$.
\end{enumerate}
\end{definition}

As evidenced by Example \ref{ex:4}, if $\bar\pi$ is a witness to $(\pi_0,\pi_1)$ satisfying EB, then $\bar\pi(S_1)$ is bounded above by the ratio $\frac{\pi_0(F)}{\pi_1(F)}$ for events where this is well defined (i.e., where $F \in \S_0$ and $\pi_1(F)>0$). Thus, \eqref{c3} is a technical assurance that this ratio does not vanish, as this would preclude the existence of a witness placing positive probability on the conditioning event. 

Note that while \eqref{c3} is a technical continuity type condition, in many circumstances it can be verified easily:
\begin{rmk}
\label{rmk:inf}
If $S_1$ is not completely non-measurable then \eqref{c3} is implied by commensurability and in particular 
$$0 < \frac{\pi_0(E)}{\pi_1(E)} = \inf\limits_{\substack{F \in \S_0, \\ \pi_1(F)>0}}\frac{\pi_0(F)}{\pi_1(F)}.$$ 
To see this, note that there exists a non-empty $E \in \Sigma_0$ with $E \subseteq S_1$, and
for such $E$, $\pi_1(E) > 0$ by definition, so by \eqref{c1}, $\pi_0(E) >0$ as well. So $0 < \frac{\pi_0(E)}{\pi_1(E)}$. Moreover, for any other $F \in \Sigma_0$ with  $\pi_1(F) > 0$, \eqref{c2} provides $\pi_0(E)\pi_1(F) \leq \pi_1(E)\pi_0(F)$ which we can rearrange to obtain
$0 < \frac{\pi_0(E)}{\pi_1(E)} \leq \frac{\pi_0(F)}{\pi_1(F)}$.
\end{rmk}

The following theorem, a clear extension of Theorem \ref{prop:update}, shows that this boundedness is exactly the additional requirement to capture EB with countable state spaces.

\begin{thm}
\label{prop:updateinf}
Let $\W$ be at most countable: then $\pi_1$ is boundedly commensurate to $\pi_0$ if and only if $(\pi_0,\pi_1)$ satisfies EB and $\pi_0^{*}(S_1) > 0$.
\end{thm}

\begin{proof}
The `if' direction is easy: Assume $\pi_0^{*}(S_1) > 0$ and that $(\pi_0,\pi_1)$ satisfies EB with $\bar\pi$ a measure witnessing this fact. \eqref{c1} is obvious. Towards, \eqref{c2}: Take some $E,F \in \Sigma_0$ with $E \subseteq S_1$. If $\pi_1(F) = 0$ then \eqref{c2} holds immediately, so assume $\pi_1(F) > 0$. Then, by the properties of $\bar\pi$, 
$$\frac{\pi_0(E)}{\pi_0(F)} = \frac{\bar\pi(E)}{\bar\pi(F)} 
 = \frac{\bar\pi(E \cap S_1)}{\bar\pi(F)}
  \leq \frac{\bar\pi(E\cap S_1)}{\bar\pi(F\cap S_1)} = \frac{\pi_1(E)}{\pi_1(F)}$$
establishing \eqref{c2}. 
\eqref{c3} holds because
$$\frac{\pi_0(F)}{\pi_1(F)} = \bar\pi(S_1)\frac{\pi_0(F)}{\bar\pi(F \cap S_1)} \geq  \bar\pi(S_1)\frac{\pi_0(F)}{\bar\pi(F)}=  \bar\pi(S_1)\frac{\pi_0(F)}{\pi_0(F)} = \bar\pi(S_1).$$ 
for all $F \in \Sigma_0$ with $\pi_1(F) > 0$.

Towards the `only if' direction, assume that $\pi_1$ is commensurate to $\pi_0$.
We must find a $\bar\pi$ on $(\W,\Sigma_1)$ such that the conditions of EB hold.

First, we must set a value, $\beta$, for $\bar\pi(S_1)$. If there exists an $E \in \Sigma_0$ with $E \subseteq S_1$, then set $\beta = \frac{\pi_0(E)}{\pi_1(E)}$. By Remark \ref{rmk:size}, the choice of $E$ is irrelevant and $\beta \leq 1$. Further, \eqref{c1}  indicates that $0 < \beta$ and Remark \ref{rmk:inf} that $\beta \leq \inf\limits_{\substack{F \in \S_0, \\ \pi_1(F)>0}} \frac{\pi_0(F)}{\pi_1(F)}$. If no such $E$ exists, take an arbitrary $0< \beta \leq  \inf\limits_{\substack{F \in \S_0, \\ \pi_1(F)>0}} \frac{\pi_0(F)}{\pi_1(F)} \leq 1$. If $\beta = 1$ then setting $\bar\pi = \pi_1$ suffices, so assume $\beta \in (0,1)$.

Define the $\sigma$-algebras $\S^\bullet = \{ E \in \S_1 \mid E \subseteq S_1 \}$ and $\S^\circ_0 = \{F \cap S_1^c \mid F \in \S_0 \}$, $\S^\circ = \{ E \in \S_1 \mid E \subseteq S_1^c \}$. Let $\pi^\bullet: \S^\bullet \to [0,1]$ defined by $\pi^\bullet(E) = \beta \pi_1(E)$. Moreover, let $\pi^\circ_0: \S^\circ_0 \to [0,1]$ defined by $\pi^\circ_0(F \cap S_1^c) = \pi_0(F) - \pi^\bullet(F \cap S_1)$. 

To see that this is well defined, let $F,F' \in \S_0$, with $F\neq F'$ and $F \cap S_1^c = F' \cap S_1^c$. We will consider the case where $F' \subset F$, but the other cases clearly follow from the same argument. So we have:  $\emptyset \neq F \setminus F' \subseteq S_1$ and $F \setminus F' \in \S_0$. Then
\begin{align*}
\pi_0(F) - \pi^\bullet(F \cap S_1) &= \pi_0(F') + \pi_0(F \setminus F') - \pi^\bullet(F' \cap S_1) - \pi^\bullet(F \setminus F') \\
&= \pi_0(F') + \pi_0(F \setminus F') - \pi^\bullet(F' \cap S_1) - \frac{\pi_0(F \setminus F')}{\pi_1(F \setminus F')} \pi_1(F \setminus F') \\
&= \pi_0(F') - \pi^\bullet(F' \cap S_1)
\end{align*}
That $\pi^\circ_0$ is non-negative follows from the observation that $\pi_0(F) \geq \pi_1(F)\beta$ for all $F \in \S_0$. Moreover, it is clear that $\pi^\circ_0$ is additive and vanishes on the null set, and hence is a measure. 

Now let $\pi^\circ$ be an arbitrary extension of $\pi^\circ_0$ to $\Sigma^\circ$ (which exists because $\W$ is countable. This is the step in the proof that causes issue when extending to an uncountable state-space). Finally, set $\bar\pi: E \mapsto \pi^\bullet(E \cap S_1) + \pi^\circ(E \cap S_1^c)$. To verify \eqref{eb2} notice that for $F \in \S_0$, $\bar\pi(F) = \pi^\bullet(F \cap S_1) + \pi^\circ(F \cap S_1^c) = \pi^\bullet(E \cap S_1) + \pi_0(F) - \pi^\bullet(F \cap S_1) = \pi_0(F)$. (Moreover, this verifies also that $\bar\pi$ is a probability measure). To see \eqref{eb3} notice first that $\bar\pi(S_1) = \beta$ and so
$$ \pi_1(E) = \pi_1(E \cap S_1) = \frac{\pi^\bullet(E \cap S_1)}{\beta} = \frac{\bar\pi(E \cap S_1)}{\bar\pi(S_1)}$$
for $E \in \S_1$.
%
\end{proof}

When trying to move from countable to arbitrary state-spaces two issues arise. The first, surmountable, issue is that care needs to taken as the support of a measure is no longer well defined (without, additional---e.g., topological---restrictions) so conditioning events are identified only up to sets of measure zero. This problem can be dealt with by beginning with two topologies $\tau_0$ and $\tau_1$ (where $\tau_0$ is coarser than $\tau_1$) and assuming that $\Sigma_0$ and $\Sigma_1$ are the respective Borel $\sigma$-algebras. Then, the support $S_1$ is defined per-usual: the smallest closed subset of $\W$ for which every open neighborhood of every point of $S_1$ has $\pi_1$ positive measure. 

The second is more subtle and intractable: While $\pi_1$ completely determines $\bar\pi$ over $S_1$, the extension of $\pi_0$ over $S_1^c$ is essentially unrestricted.\footnote{This is seen in the proof of Theorem \ref{prop:updateinf} when $\pi^\circ$ is taken to be an arbitrary extension of $\pi^\circ_0$.} This seeming flexibility in the choice of $\bar\pi$ can in fact cause problems: in arbitrary state-spaces, not every measure can be extended to a finer $\sigma$-algebra. For example, let $\S_0$ be the (completion of the) Borel $\sigma$-algebra on $[0,1]$ and let $\S_1$ be the powerset of $[0,1]$; let $\pi_0=\frac12 \delta_0 + \frac12 \lambda$ and $\pi_1=\delta_0$ where $\delta_0$ is the Dirac measure on $0$ and $\lambda$ is the Lebesgue measure on $[0,1]$. While is it intuitively clear that $\pi$ is derived from $\pi_0$ by updating on $S_1 = \{0\}$, assuming the axiom of choice, there does not exist  \emph{any} extension of $\pi_0$ to $\Sigma_1$, as this would contradict the existence of non-Lebesgue-measurable sets. 

In general: if for an arbitrary state spaces there exists an extension of $\pi_0$ to $\Sigma_1$ then Theorem \ref{prop:updateinf} holds and the proof here stated goes through without issue. In particular, if $\Sigma_1$ is generated by $\Sigma_0$ and \emph{finitely} many novel events then an extension exists. This covers the natural case where the DM becomes aware of, then learns, a single new distinction (i.e., can distinguish between $E$ and $E^c$, and then learns which one is true). But, the existence of an extension can fail even when $\Sigma_1$ is generated by $\Sigma_0$ and countably many new events \citep{ershov1975extension}. Although necessary and sufficient  conditions exist on the primitives to ensure the existence of extensions (see again \cite{ershov1975extension}), these conditions, as far as I can tell, have no economic interpretation and I omit discussing them here.

\section{Repeated Conditioning}
\label{sec:rep} If the DM discovers unforeseen evidence more than once, the observed subjective probabilities will form a finite sequence, $\pi_0 \ldots \pi_N$, over increasingly fine algebras, $\Sigma_0 \ldots \Sigma_N$. Let $\pi_i^{*}$ be defined over $\Sigma_{i+1}$ for $i < N$. For a DM who adheres to Bayesianism to the extent possible under unawareness, each $(\pi_n,\pi_{n+1})$ will satisfy EB. Even if the modeler cannot feasibly observe each $\pi_n$, this hypothesis can be falsified, since under this assumption, $(\pi_n,\pi_{m})$ will satisfy EB whenever $m \geq n$. 

It is also possible to find a \emph{single} $\bar\pi$ that acts as the common extension for each period's beliefs. That is, there exists a $\bar\pi$ such that $\bar\pi(\wc \mid S_n)$ is an extension of $\pi_n$ to $S_n$. In words, if we look at the restriction of $\bar\pi$ to $\Sigma_n$, then $\pi_n$ is formed by conditioning this restriction on $S_n$---the $n^{th}$ period's evidence. 

\begin{ex}
Let $\W = \{\w_1, \ldots \w_5\}$. Let $\S_0$ be generated by the partition $\{\{\w_1,\w_2,\w_3\},\{\w_4,\w_5\}\}$ and let $\pi_0$ assign the two cells $\frac12$. Let $\S_1$ be generated by the partition $ \{\{\w_1,\w_2\},\{\w_3\},\{\w_4,\w_5\}\}$ and let $\pi_1$ assign the three cells 0, $\frac13$ and $\frac23$, respectively.  Finally, let $\S_2$ be generated by the discrete partition, and set $\pi_2(\w_1) = \pi_2(\w_2) = \pi_2(\w_5) = 0$ and $\pi_2(\w_3) = \pi_2(\w_4) = \frac12$.

Notice that both $(\pi_0,\pi_1)$ and $(\pi_1,\pi_2)$ satisfy EB---so too does $(\pi_0,\pi_2)$. Moreover, consider $\bar\pi$ over $\S_2$ given by $\bar\pi(\w_1) = \bar\pi(\w_2) = \frac18$ and $\bar\pi(\w_3) = \bar\pi(\w_4) = \bar\pi(\w_5) = \frac14$. Notice that $\pi_0$ is the restriction of $\bar\pi$ to $\S_0$, and $\pi_1$ the restriction to $\S_1$ conditional on $S_1$ and $\pi_2$ the (trivial) restriction to $\S_2$ conditional of $S_2$.
\end{ex}

\begin{thm}
\label{thm:rep}
Let $(\pi_0 \ldots \pi_{N})$ be defined over increasingly fine algebras, $\Sigma_0 \ldots \Sigma_N$ (all over a common state space $\W$). If $(\pi_n,\pi_{n+1})$ satisfies EB and $\pi_{i}^{*}(S_{i+1}) > 0$, for each $n \in \{0,\ldots ,N-1\}$, then $(\pi_n,\pi_{m})$ will satisfy EB for all $0 \leq n \leq m \leq N$. Moreover, in such cases, there exists a $\bar\pi \in (\W,\S_N)$ for all $n \leq N$, $\bar\pi(\wc \mid S_n)$ is an extension of $\pi_n$ to $S_n$.
\end{thm}

\begin{proof}
We will show this for $N = 2$, the general case following easily by induction. Let $\bar\pi_{01}$ and $\bar\pi_{12}$ denote the two witnesses to $(\pi_0,\pi_1)$ and $(\pi_1,\pi_2)$ satisfying expected Bayesianism, respectively. We will show that $\bar\pi_{01}$ is  boundedly commensurate to $\pi_{2}$. This suffices to prove the claim by appealing to Theorem \ref{prop:update} as it allows us to find a witness, $\bar\pi_{02}$, to $(\bar\pi_{01}, \pi_2)$ satisfying EB. Indeed, since $\bar\pi_{01}$ extends $\pi_0$ so too does $\bar\pi_{02}$ which therefore serves as a witness to $(\pi_0, \pi_2)$ satisfying EB. Further, since $\bar\pi_{02}$ extends $\bar\pi_{01}$, we have that, for all $E \in \Sigma_1$ with $E \subseteq S_1$: $\pi_1(E) = \frac{\bar\pi_{01}(E)}{\bar\pi_{01}(S_1)} = \frac{\bar\pi_{02}(E)}{\bar\pi_{02}(S_1)} =  \bar\pi_{02}(E \mid S_1)$. 

So fix some $E,F \in \S_1$ such that $E \subseteq S_2$. Let $\bar\pi_{01}(F) = 0$, then $\pi_1(F) = \frac{\bar\pi_{01}(F\cap S_1)}{\bar\pi_{01}(S_1)} = 0$ so by the commensurability of $\pi_1$ and $\pi_2$, $\pi_2(E) = 0$. \eqref{c1} holds. Moreover, we have:
\begin{align*}
\bar\pi_{01}(E)\pi_2(F) &= \bar\pi_{01}(S_1)\frac{\bar\pi_{01}(E\cap S_1)}{\bar\pi_{01}(S_1)}\pi_2(F) &&\text{ (since $S_2 \subseteq S_1$)} \\
&= \bar\pi_{01}(S_1) \pi_1(E)\pi_2(F) &&\text{ (by definition of $\bar\pi_{01}$)}\\
&\leq \bar\pi_{01}(S_1) \pi_2(E)\pi_1(F) &&\text{ (by \eqref{c2} for $\pi_1$ and $\pi_2$)}\\
&= \bar\pi_{01}(S_1) \pi_2(E)\frac{\bar\pi_{01}(F\cap S_1)}{\bar\pi_{01}(S_1)}  &&\text{ (by definition of $\bar\pi_{01}$ again)}\\
&\leq \pi_2(E)\bar\pi_{01}(F),
\end{align*}
so \eqref{c2} holds. Finally, 
$$ \inf\limits_{\substack{F \in \S_1, \\ \pi_2(F)>0}} \frac{\bar\pi_{01}(F)}{\pi_2(F)} \geq \inf\limits_{\substack{F \in \S_1, \\ \pi_2(F)>0}} \frac{\bar\pi_{01}(F \cap S_1)}{\pi_2(F)} =  \bar\pi_{01}(S_1)\inf\limits_{\substack{F \in \S_1, \\ \pi_2(F)>0}} \frac{\pi_{1}(F)}{\pi_2(F)} > 0,$$
so \eqref{c3} holds. This completes the proof.
\end{proof}

Notice that if we drop the requirement that $\pi_{i}^{*}(S_{i+1}) > 0$ the result is no longer true and in fact is not even true for standard Bayesianism. For example let $\W = \{\w_1,\w_2,\w_3\}$, $\S_0 = \S_1 = \S_2$ be generated by the discrete partition, and  
\begin{align*}
\pi_0: \begin{cases}
\w_1 \mapsto 0 \\
\w_2 \mapsto \frac13 \\
\w_3 \mapsto \frac23 
\end{cases}
\quad
\pi_1: \begin{cases}
\w_1 \mapsto 1 \\
\w_2 \mapsto 0 \\
\w_3 \mapsto 0 
\end{cases}
\quad
\pi_2: \begin{cases}
\w_1 \mapsto 0 \\
\w_2 \mapsto \frac23 \\
\w_3 \mapsto \frac13
\end{cases}.
\end{align*}
Although each successive pair satisfies EB (by virtue of updating on an event of ex-ante measure 0), $(\pi_0,\pi_2)$ does not satisfy EB (or Bayesianism for that matter).

\section{A Few Notes on Related Literature} 
\label{sec:lit}

The tenet of \emph{reverse Bayesianism} (RB), as introduced by \cite{karni2013reverse}, states that when the DM becomes more aware, her relative probabilistic assessments regarding previously understood contingencies do not change. Formally, reverse Bayesianism (RB) \cite{karni2013reverse,karni2017awareness} is captured by the restriction
$\frac{\pi_1(E)}{\pi_1(F)} = \frac{\pi_0(E)}{\pi_0(F)}$
 for all $E,F \in \Sigma_0$. In the present context, by setting $F = \W$, we see this implies $\pi_1|_{\Sigma_0} = \pi_0$, so that $\pi_1$ is an extension of $\pi_0$ to the richer algebra. EB generalizes RB: the transition ($\pi_0 \to \pi_1$) satisfies reverse Bayesianism if and only if it satisfies EB and $S_1=\W$ so that the conditioning step is trivial.  
Even if becoming aware does not intrinsically change beliefs, it may well be that by the time the DM's beliefs can actually be elicited, she has taken into account some additional probabilistic information. That is to say, despite the DM adhering to RB, the beliefs elicited at time 1 reflect not only the expansion of awareness but also conventional updating. 

Of course, there are also many intuitive situations where becoming aware intrinsically does provide information. Incontrovertibly, if the DM becomes aware of an event $E$, she \emph{must} learn that she used to be unaware of $E$.\footnote{In purely semantic ``state-space" models, introspection is not captured. However, by starting with a first order language with an awareness modality and setting the states as possible worlds, one can make precise sense out of the event ``$i$ used to be unaware of the event E." See, for example, \cite{halpern2009reasoning, halpern2020dynamic} and appendix \ref{sec:DT2}.} But, even without appealing to introspection, it is reasonable that the mere existence of a concept can serve as evidence regarding contingencies the DM was already aware of. This is essentially the ``problem of old evidence'' \citep{glymour1980theory}. 

\begin{ex}
Players $i$ and $j$ are playing a card game. $i$ initially thinks it is highly likely that he and $j$ fully understand the rules of the game, and further that $j$'s behavior is not rationalizable according to these rules. Hence $i$ believes it is highly likely that $j$ is irrational. $i$ then discovers that there are in fact two variants of the game. Although $i$ does not learn any hard information about the rules of either variant, he now considers it much more likely that $j$ is best responding  (to the rules of the game $j$ believes they are playing), and therefore $i$ places less probability on the event that $j$ is irrational. 
\hfill $\blacksquare$
\end{ex}
 
 In \cite{karni2013reverse}, there are actually two distinct ways the DM can become more aware: \emph{refinement}, which is essentially what is characterized here (where $\Sigma_1$ is richer than, but defined on the same space as, $\Sigma_0$), and \emph{expansion} where the underlying state-space gets larger (so that $\Sigma_0$ is defined on $\W$ and $\Sigma_1$ on $\W\cup\W'$). Under expansions RB does not imply that $\pi_1$ is an extension of $\pi_0$, but rather that $\pi_1$ \emph{conditional on $\W$} is an extension. 

One can represent expansions via refinements by setting an event $E^\star \in \Sigma_0$ to collect ``that which is not yet understood." $E^\star$ gets carved up with each new discovery. This latter method has the added benefit of allowing the DM to reason about her own unawareness. If, however, we insist on entertaining expansions of the state-space itself so that $\Sigma_0$ is defined on $\W$ and $\Sigma_1$ on $\W\cup\W'$, then we can appropriately generalize the definition of extended Bayesianism to allow $\pi_1$ to entertain probability on newly discovered states: Setting $\pi_1 \in \D(\W\cup\W', \Sigma_1)$, say $(\pi_0,\pi_1)$ satisfies \emph{generalized extended Bayesianism} (GEB) if $\pi_1(\W) > 0$ and $(\pi_0,\pi_1(\cdot \mid \W))$ satisfy EB. In this case, we have that the overall transition ($\pi_0 \to \pi_1$) satisfies reverse Bayesianism if and only if $\W \subseteq S_1$. 

\cite{karni2018reverse} consider the case where a DM, in the process of becoming more aware, might simultaneously condition her beliefs with respect to some event, $E$. They only consider expansions of the state space and not refinements of previously describable events (i.e., $\W$ expands to $\W \cup\W'$ but $\Sigma_0 = \{E \cap \W \mid E \in \Sigma_1\}$).  They introduce  \emph{generalized reverse Bayesianism}, whereby the relative probabilities of events must remain the same only for events in $S_0 \cap S_1$ (rather than all of $S_0$ as is the case for RB). This case is clearly captured by GEB. The overall transition ($\pi_0 \to \pi_1$), where $\pi_1$ is defined on $(\Sigma_1,\W \cup \W')$, satisfies generalized reverse Bayesianism if and only if $(\pi_0,\pi_1)$ satisfy GEB and $\Sigma_0 = \{E \cap \W \mid E \in \Sigma_1\}$.

\cite{fagin1990new} introduced the notion of outer and inner conditional probability as the upper and lower envelopes of the conditional probabilities of all possible extensions to a richer algebra. In the language of this paper, the outer conditional probability of $\pi_0$ on $E \in \Sigma_1$ is
$$\pi_{0}^*(\cdot | E) = \sup\{ \bar\pi(\cdot | E) \mid \bar\pi \in \D(\W, \S_1), \bar\pi \text{ extends } \pi_0\}$$
and the inner conditional probability is defined by replacing the $\sup$ with an $\inf$. Thus it must be that $(\pi_0,\pi_1)$ satisfies EB exactly when $\pi_1$ lies inside of the outer and inner conditional probabilities of $\pi_0$ (where the conditioning event is $S_1$). As such, filtering through inner and outer probability provides another, indirect, characterization of unforeseen posteriors. 

\cite{manski1981learning} proposed a model of sequential refinement, wherein a decision maker can \emph{choose} to refine her sigma algebra at a cost. Because precise probabilities cannot be calculated on not-yet-measurable events, the agent's decision to refine her sigma algebra is governed by minimax rules.


\cite{halpern2020dynamic} consider an agent who simultaneously becomes more aware and conditions her beliefs. In complement to this paper, they consider a syntactic logical model, where knowledge and awareness are represented directly by statements in a formal language, and the notion of updating is captured by a transition rule on semantic models.

There is a large body of literature from different fields on how to update beliefs in response to \emph{unexpected} (i.e., ex-ante probability 0) rather than \emph{unforseen} (i.e., non-measurable) evidence. Some proposals: Using a sequence of probability measures with disjoint supports (lexicographic probability systems; \cite{blume1991lexicographic}); taking as the primitive a family conditional probabilities conditioned on the algebra of events (conditional probability systems;  \cite{renyi1955new}); developing a measure theory using non-standard analysis with infinitesimals \citep{robinson1973function}. For an overview of the relations between these approaches, see \cite{halpern2010lexicographic}. Recently, in the economics literature there have been models of non-Bayesian updating when the evidence is surprising (i.e., sufficiently close to ex-ante probability 0) which allows for the capture of paradigm shifts \citep{ortoleva2012modeling, galperti2019persuasion}. These approaches concern a different setting than the present paper where the DM might learn an event she had \emph{no} probabilistic assessment of. Notice that in stark contrast to a measure 0 event, an unforeseen event has also an unforeseen complement; this recasts the well discussed distinction between being unaware of an event and believing the event impossible (for example see \cite{modica1999unawareness}).

\bibliographystyle{plainnatnourl.bst}
\singlespace
\bibliography{DA.bib}

\appendix 

\section{Decision Theory under Unawareness}
\label{sec:DT2}

\def\t{\textsc{t}}
\def\f{\textsc{f}}
\def\p{\textsc{p}}
\def\q{\textsc{q}}

\newcommand{\propscript}[1]{\mathbb{#1}}

\def \P{\propscript{P}}
\def \X{\propscript{X}}
\def \A{\propscript{A}}
\def \S{\propscript{S}}

\def \t{\textsc{t}}

\def \T{\mathbf{T}}
\def \F{\mathbf{F}}

The decision theoretic treatment in Section \ref{sec:DT} might appear to lack a consistent interpretation---the acts considered by the DM are mathematical functions defined on a state space the DM might not fully understand. Worse still, if the DM is unaware of certain aspects of the decision problem, then perhaps her understanding of an event $E$ depends on how the event is described to the DM: both ``It is raining or it is not raining" and ``The axiom of choice implies that every vector space has a basis" are tautologies, so refer to the same event, but a DM unaware of set theory will in general treat these statements differentially. 

This section outlines a foundational approach to assessing a DM's probabilistic assessments under unawareness, following a syntactic approach whereby the objects of choice are language based. In this model, it is easy to separate the DM's awareness (she understands a subset of the language) from her probabilistic assessments of uncertainty (the likelihood she assigns to statements  being true, given she is aware of them). 

The uncertainty faced by the DM is captured by a set of propositions, $\P$, each of which can be either true or false. These are statements about the world ``A quantum computer can factor integers in polynomial time" or ``The marginal cost of production is constant,'' and we can think of them as verbal descriptions that must be interpreted by the agent. 

$\P$ contains two distinguished propositions $\T$ and $\F$, that are interpreted as ``true'' and ``false'' respectively. Then, $\L(\P)$ is the language defined inductively, beginning with $\P$ and such that if $\phi,\psi$ are in $\L(\P)$ then so too are $\neg \phi$ and $(\phi \land \psi)$. The interpretation is as in propositional logic: $\neg \phi$, the \emph{negation} of $\phi$, is interpreted as the statement that $\phi$ is not true and $\phi\land\psi$, the \emph{conjunction} of $\phi$ and $\psi$, is interpreted as the statement that both $\phi$ and $\psi$ are true.

At a given time the DM will not be aware of all propositions, but rather a subset $\A_t \subseteq \P$. Our assumption that awareness expands indicates that $A_0 \subseteq A_1$. Thus her understanding of the uncertainty she faces is limited by her unawareness.

If $\W$ is a state space, then call $\t: \L(\P) \to 2^\W$ a \emph{truth valuation} if it obeys the rules of logic: $\t(\T) = \W$ and $\t(\F) = \emptyset$; $\t(\neg \phi) = \t(\phi)^c$; $\t(\phi \land \psi) = \t(\phi) \cap \t(\psi)$. The interpretation is that $\t(\phi)$ is the set of states where $\phi$ is true. Unlike the model provided above, the state space and truth valuation is not given exogenously, but will be part of the representation.

The primitive here is a time-indexed preference relation over \emph{bets}. A \emph{bet} is a pair $(\phi,x) \in \L(\P) \times X$  (where $X$ is a convex consumption space with worst element $\w$, as in Section \ref{sec:DT}) written $x_\phi$. The interpretation of $x_\phi$ is a bet that pays the prize $x$ when $\phi$ is true and provides $\w$ otherwise. Notice that if $(\W,\t)$ is some commonly known state space and truth valuation, then $x_\phi$ corresponds to the state space based act $x_{\t(\phi)}$. Let $\B_t$ denote the set of bets of the form $x_\phi$ with $\phi \in \L(\A_t)$ so that $\B_0 \subseteq \B_1$.

Given a language $\L(\P)$, a subjective model of uncertainty is a quadruple: $M = (\W,\t,\l, u)$ where $\W$ is a state space, $\t$ is a truth valuation $\t: \L(\P) \to 2^\W$, $\pi$ is a probability distribution over $(\W,\Sigma)$ (with $\Sigma$ a sigma-algebra rich enough to make $\t$ measurable) and $u: X \to \R_+$ is a utility index with $u(w) = 0$.
 Say that $M$ \emph{represents} $\s_t$ if, for all pairs of bets,
$x_\phi \s_t y_\psi$ if and only if 
$$ 
u(x) \l(\t(\phi)) \geq u(y) \l(\t(\phi)).$$ 
\cite{piermont2020failures} provide conditions on the primitive $\s$ to ensure its representation by a subjective model. 

From here, we can incorporate analogs of the usual decision theoretic restrictions: call $\phi \in \A_t$ $t$-null if $x_\phi \sim_t y_\phi$ for all $x,y\in X$. Call a statement $\phi$ \emph{discarded} if it is $1$-null and for any $\psi$ such that $\phi \Rightarrow \psi$ (where $\Rightarrow$ is deduction under the rules of propositional logic) then $\psi$ is not 0-null. 

\begin{definition}
Call $(\s_0,\s_1)$ \emph{extension consistent} if for all $x,y \in X$, we have
$x_\phi \s_0 y_{\phi'}$ and $y_{\phi'} \succ_1 x_\phi$ implies there exists some $\psi$ such that $\psi \Rightarrow \phi$ and $\psi$ is discarded. 
\end{definition}

\begin{thm}
\label{thm:syntax}
Assume  $(\s_0,\s_1)$ are extension consistent and can be represented by subjective models of uncertainty representing using a common utility index. Then there exists a common state space $\W$, a truth valuation $\t$ and a utility $u$ such that 
$(\W, \t, \pi_t, u)$ represents $\s_t$,  where $\pi_t \in (\W,\t(\A_t))$ and $(\pi_0,\pi_1)$ satisfies extended Bayesianism. 
\end{thm}

Theorem \ref{thm:syntax} follows more or less directly from Theorem \ref{thm:maineq} once we have found a common state space for the two stages of preference. Of course, this is always possible by taking a concrete representation of the Lindenbaum---Tarski algebra of $\L(\P)$. For details, see \cite{piermont2020failures}.

\end{document}